\documentclass[a4paper,11pt]{article}

\usepackage{graphicx}

\usepackage{fullpage}
\usepackage{libertine}
\usepackage{color}

\usepackage[ruled,vlined]{algorithm2e}
\SetArgSty{textrm}

\usepackage{amsmath,amsfonts,amssymb,amsthm}

\usepackage[breaklinks=true]{hyperref}
\usepackage[svgnames]{xcolor}
\usepackage[capitalise,nameinlink]{cleveref}
\hypersetup{colorlinks={true},linkcolor={DarkBlue},citecolor=[named]{DarkGreen}}

\usepackage{natbib}

\usepackage{subcaption}
\usepackage{caption}
\newcommand{\alphaMech}{\text{\sc $\alpha$-Statistic}}

\newcommand{\UniformAlpha}{\text{\sc 
Uniform-Statistic}}

\usepackage{tikz}  
\usetikzlibrary{arrows}
\usetikzlibrary{patterns,snakes}
\usetikzlibrary{decorations.shapes}
\tikzstyle{overbrace text style}=[font=\tiny, above, pos=.5, yshift=5pt]
\tikzstyle{overbrace style}=[decorate,decoration={brace,raise=5pt,amplitude=3pt}]
\usetikzlibrary{shapes.geometric}

\newtheorem{theorem}{Theorem}[section]

\theoremstyle{definition}
\newtheorem*{comment*}{Comment}

\newcommand{\SW}{\text{SW}}

\newcommand{\bw}{\mathbf{w}}
\newcommand{\bo}{\mathbf{o}}
\newcommand{\by}{\mathbf{y}}
\newcommand{\bx}{\mathbf{x}}
\newcommand{\bp}{\mathbf{p}}

\title{\bf Constrained Truthful Obnoxious Two-Facility Location with Optional Preferences}

\usepackage{authblk}
\author{Panagiotis Kanellopoulos and Alexandros A. Voudouris}
\date{School of Computer Science and Electronic Engineering \\ University of Essex, UK}

\begin{document}

\allowdisplaybreaks

\maketitle

\begin{abstract}
We consider a truthful facility location problem with agents that have private positions on the line of real numbers and known optional preferences over two obnoxious facilities that must be placed at locations chosen from a given set of candidate ones. Each agent wants to be as far away as possible from the facilities that affect her, and our goal is to design mechanisms that decide where to place the facilities so as to maximize the total happiness of the agents as well as provide the right incentives to them to truthfully report their positions. We consider separately the setting in which all agents are affected by both facilities (i.e., they have non-optional preferences) and the general optional setting. We show tight bounds on the approximation ratio of deterministic strategyproof mechanisms for both settings, and almost tight bounds for randomized mechanisms. 
\end{abstract}

\section{Introduction}
Due to their versatility in modeling many applications of interest, such as {\em voting} and {\em clustering}, {\em facility location problems} have been a topic of multidisciplinary research. In such problems, the objective is to decide the location of a number of facilities under various constraints and typically aiming to minimize the transportation cost of a set of agents. In the economics literature, such problems have mainly been studied under the lens of {\em mechanism design}, where the agents are assumed to be strategic, and the goal is to define strategyproof decision-making rules (mechanisms) that incentivize the agents to report their private positions {\em truthfully}~\citep{Moulin1980}. On the other hand, the operations research and computer science literature have traditionally focused on versions of the problem where the positions of the agents are assumed to be known and the goal is design (approximation) algorithms to minimize objective functions of the distances of the agents to the computed facility locations~\citep{approx-book}.

Motivated by these two lines of research, \citet{procaccia09approximate} initiated the study of {\em truthful facility location problems} in which the goal is to achieve strategyproofness and approximate efficiency {\em simultaneously}. Since their work, many different such problems have been studied in the more general area of {\em approximate mechanism design without money}, under different assumptions, such as about the number of facilities, whether the location space is constrained, or whether the facilities are considered desirable or {\em obnoxious}. We refer to the survey of \citet{fl-survey} for an exposition of relevant models. 

In this paper we focus on a discrete heterogeneous truthful facility location problem in which there are {\em two} obnoxious facilities to place at locations chosen from a given multiset of candidate ones. The agents are assumed to have private positions on the line of real numbers and known {\em optional preferences} over the two facilities, in the sense that some of them might want to be far from one of the two facilities (in which case those agents derive a {\em utility} that is equal to their distance from the facility that affects them), while some other agents might want to be far from both (and hence derive utility that is equal to their distance from both facilities). We want to design strategyproof mechanisms and approximately maximize the {\em social welfare}, that is, the total utility of the agents. To give a toy example of our setting, consider the case where the facilities correspond to a factory and a bar. Generally, both of these facilities can be considered obnoxious by the agents since the factory is not good health-wise and the bar can get quite noisy, but some of the agents might be indifferent if, say, the bar is located close to their houses or not. 

\subsection{Our Contribution}
We show tight or nearly tight bounds on the approximation ratio of {\em deterministic} and {\em randomized} strategyproof mechanisms for the truthful obnoxious two-facility location problem that we consider. We provide details about the model in Section~\ref{sec:model}. Our results are summarized in Table~\ref{tab:overview}.

\renewcommand{\arraystretch}{1.3}
\begin{table}[t]
    \centering
    \begin{tabular}{c|cc}
                      & Deterministic & Randomized \\ 
       \hline
       Non-optional   & $\sqrt{3}$ & $[1.5, 1.522]$ \\
       General      & $3$ & $[1.5,2]$ \\
       \hline
    \end{tabular}
    \caption{An overview of our bounds on the approximation ratio of deterministic and randomized mechanisms for instances with non-optional preferences (in which agents are affected by both facilities) and general instances (in which some agents might be affected by a single facility and some agents might be affected by both). When a single number is given, the bound is tight.}
    \label{tab:overview}
\end{table}

We first consider a simpler setting in which the agents have {\em non-optional preferences} in the sense that all agents are affected by both facilities. We show a tight bound of $\sqrt{3}$ on the approximation ratio of deterministic mechanisms. This resolves the gap between $3/2$ and $2$ left from the previous works of \cite{GaiLW22} and \citet{ZhaoLNF24} who focused on this particular setting. Our upper bound follows by a mechanism that is parameterized by an $\alpha \in [0,1/2]$ and bases its decision on whether the $\alpha n$-th and the $(1-\alpha)n$-th leftmost agents agree on which between the leftmost and the rightmost candidate locations is the most-distant one. Our lower bound is based on a sophisticated construction that involves a number of sequences of instances to argue about the behavior of any deterministic mechanism. See Section~\ref{sec:doubleton:deterministic} for details. For the class of randomized mechanisms, we show that the approximation ratio is between $3/2$ and $(5+4\sqrt{2})/7 \approx 1.522$. The upper bound here follows by running the deterministic parameterized mechanism for an $\alpha$ that is chosen uniformly at random; see Section~\ref{sec:doubleton:randomized}.

We next switch to the general setting, in which some agents might be affected by one of the facilities and some agents might be affected by both. For deterministic mechanisms, we show a tight approximation ratio of $3$. The upper bound follows by a mechanism which always places the two facilities at the leftmost and rightmost candidate locations. It decides which facility goes where by computing, for each facility, which of these two locations the majority of agents affected by the facility prefer, and then gives priority to the facility that has the strongest majority. See Section~\ref{sec:general:deterministic}. We finally reveal a dichotomy between deterministic and randomized mechanisms by showing an improved bound of $2$ using a rather simple randomized mechanism which equiprobably places every facility at each of the leftmost and rightmost locations. 

\subsection{Related Work}
As already mentioned above, \citet{procaccia09approximate} were the first to design strategyproof mechanisms that are approximately efficient for truthful facility location problems with desirable facilities. In particular, they studied the case of a single facility, the case of two facilities, and the case of agents that might have multiple positions. All these variants and many more have been studied further in the literature; see \citep{fl-survey} for a detailed exposition of the many different models and techniques used over the years. We discuss here some papers that are related to some aspects of our model, but are different in at least one. 

In our model, when the agents are affected by both facilities, we assume that the utility they derive is equal to the {\em sum} of distances from both facilities. This is known as the {\em heterogeneous} model and has also been considered in several previous works for the case of desirable facilities~\citep{serafino2016,Xu2021minimum,kanellopoulos2023discrete,kanellopoulos2024} as well as for obnoxious ones~\citep{GaiLW22,ZhaoLNF24,Wu2024homogeneous}. Models in which the cost or utility of an agent is determined by the closest or the farthest facility have also been considered~\citep{Lu2010two-facility,fotakis2014two,chen2020max,li2020constant,DLV24}. 

Another aspect of our model is that the agents might have different preferences over the facilities. The optional preferences we consider in particular, have also been studied in other models, mainly for desirable facilities~\citep{serafino2016,kanellopoulos2023discrete,chen2020max,li2020constant}. To the best of our knowledge, the only other paper that has considered such preferences for obnoxious facilities is that of \citet{LiPS23} who, however, focused on the case where the utility of an agent is her distance to the closest facility among those that affect her, in contrast to our model, where the utility is the sum. Other types of preferences over the facilities have also been studied, such as {\em fractional}~\citep{fong2018fractional} and {\em hybrid}~\citep{feigenbaum2015hybrid}. 

A final important aspect of our model is that it is {\em discrete} in the sense that the facilities can only be placed at specific given locations rather than anywhere on the line, which was the case in the original paper of \citet{procaccia09approximate} and follow-up work. This assumption has also been made in several works on desirable facility location models~\citep{serafino2016,kanellopoulos2023discrete,kanellopoulos2024}, as well as for obnoxious~\citep{chengYZ13,GaiLW22,ZhaoLNF24} (which, however, do not combine this with optional preferences). Another related work is that of \citet{Wu2024homogeneous} who considered an obnoxious non-optional facility location problem with a minimum distance requirement between the facilities; note that this differs from our model in which we have a set of fixed candidate locations.

\section{Model} \label{sec:model}
We consider an obnoxious facility location problem in which there is a set of $n$ {\em agents} and two {\em facilities} $F_1$ and $F_2$. Each agent $i$ has a private {\em position} $x_i$ on the line of real numbers and a publicly known {\em optional preference} $p_i = (p_{i,1},p_{i,2})$ over the two facilities indicating whether $i$ is affected by facility $F_j$ ($p_{i,j} = 1$) or not ($p_{i,j} = 0$) for each $j \in [2]$; let $\bx = (x_i)_i$ be the {\em position profile} and $\bp = (p_i)_i$ be the {\em preference profile}. The facilities can be placed at locations chosen from a multiset $A$ of {\em candidate locations}. We denote by $I = (\bx, \bp, A)$ an instance of this problem; when the instance is clear from context in our proofs, we will drop it from the notation below.

Given a {\em solution} $\by = (y_1,y_2)$ that is a pair of locations from $A$ for the two facilities ($y_j$ for $F_j$), each agent $i$ derives a {\em utility} $u_i(\by)$ defined as the {\em total distance} from the facilities that affect $i$:
\begin{align*}
    u_i(\by | I) = \sum_{j \in [2]} p_{i,j} \cdot d(i,y_j),
\end{align*}
where $d(i,y_j) = |x_i-y_j|$ is the distance between the position $x_i$ of $i$ and the location $y_j$ of $F_j$ according to $\by$. 
The {\em social welfare} of a solution $\by$ is the {\em total} utility of the agents:
\begin{align*}
    \SW(\by | I) = \sum_i u_i(\by).
\end{align*}
A solution $\by$ can also be {\em randomized} if it is a {\em probability distribution} over (deterministic) solutions. Then, we are interested in the {\em expected utility} $\mathbb{E}[u_i(\by|I)]$ of each agent $i$ and the {\em expected social welfare} $\mathbb{E}[\SW(\by|I)]$ of the solution $\by$.  

A {\em mechanism} $M$ takes as input the positions reported by the agents and, using also the available information about the preferences of the agents, computes a solution deciding where to place the two facilities; let $M(I)$ be the solution computed by $M$ when the instance formed by the input provided to it is $I$. We aim to design mechanisms that choose solutions with high (expected) social welfare and at the same time incentivize the agents to truthfully report their private positions. 

For a position profile $\bx$, let $(x_i',\bx_{-i})$ denote the position profile that is the same as $\bx$ with the only difference that the position of agent $i$ has been changed to $x_i'$. A deterministic mechanism is said to be {\em strategyproof} if no agent $i$ can increase her utility by misreporting her true position. Formally, for any two instances $I = (\bx, \bp, A)$ and $J = ((x_i', \bx_{-i}), \bp,A)$ which differ only on the position of $i$, it holds that
$u_i(M(I)|I) \geq u_i(M(J)|I)$. Similarly, a randomized mechanism is said to be {\em strategyproof in expectation} (or, simply, strategyproof) if no agent $i$ can increase her expected utility by misreporting her true position. All randomized mechanisms we consider in this paper are actually {\em universally strategyproof} as they are probability distributions over deterministic strategyproof mechanisms. 

The {\em approximation ratio} of a mechanism $M$ is the worst-case ratio over all possible instances between the maximum possible social welfare (achieved by any solution) and the (expected) social welfare of the (randomized) solution computed by $M$, that is,
\begin{align*}
    \sup_{I} \frac{\max_{\by \in A^2} \SW(\by|I) }{ \mathbb{E}[\SW(M(I)|I)] }.
\end{align*}
So, our goal is to design strategyproof mechanisms with an as small approximation ratio as possible. 

Before we conclude this section, we introduce some further notation and terminology that will be useful later in the paper. 
Given an instance, we will denote by $N_j$ the agents that are affected by facility $F_j$ for each $j \in [2]$; observe that an agent might belong to both $N_1$ and $N_2$ in case she is affected by both facilities, or to just one of those sets. We will pay particular attention to instances in which all agents are in $N_1 \cap N_2$; these instances correspond to the obnoxious heterogeneous truthful two-facility location problem with {\em non-optional preferences} that has been studied by~\citet{GaiLW22} and \citet{ZhaoLNF24}. 

Most of our mechanisms will compute solutions that involve the leftmost candidate location $L$ and the rightmost location $R$. In some cases, we will also consider the second leftmost location $\ell$, as well as the second rightmost location $r$; note that $\ell$ and $r$ might not exist in case $A$ consists of just two locations, or it might be the case that $\ell = r$ if there are three locations. We say that an agent $i$ {\em prefers} a location $x$ over another location $y$ in case $i$ is farther to $x$ than to $y$, i.e., $d(i,x) \geq d(i,y)$. Given this, we denote by $t(i)$ and $s(i)$ the most-preferred (that is, the most distant) location of $i$ and the second most-preferred (that is, the second most distant) location of agent $i$, respectively. Finally, we will also use the {\em triangle inequality}, which holds for any metric space (and thus also for the line of real numbers where agents and candidate locations are in our setting) and states that $d(x,y) \leq d(x,z) + d(z,y)$ for any points $x$, $y$ and $z$. 

\section{Non-Optional Preferences} \label{sec:doubleton}
We start with instances in which all agents are affected by both facilities, i.e., the agents have non-optional preferences. For such instances, we first show a tight bound of $\sqrt{3} \approx 1.73$ on the approximation ratio of deterministic strategyproof mechanisms, thus closing the gap between the bounds of $3/2$ and $2$ that were previously shown by \citet{GaiLW22} and \citet{ZhaoLNF24}. We also show that the approximation ratio of randomized mechanisms is between $3/2$ and $(5+4\sqrt{2})/7 \approx 1.522$. 

\subsection{Deterministic Mechanisms}\label{sec:doubleton:deterministic}
We consider the $\alphaMech$ mechanism which, for a parameter $\alpha \in [0,1/2]$, considers the distances of the $\alpha n$-th and the $(1-\alpha)n$-th leftmost agents to $L$ and $R$. If they both agree that $L$ is the most distant location, then the facilities are placed at $L$ and the available location farthest from the $\alpha n$-th leftmost agent. Similarly, if $R$ is the most distant, the facilities are placed at $R$ and the available location farthest from the $(1-\alpha)n$-th leftmost agent. Otherwise, if they disagree on which among $L$ and $R$ is the most distant, the facilities are placed at $L$ and $R$; see Mechanism~\ref{mech:sc-alpha}.

\SetCommentSty{mycommfont}
\begin{algorithm}[h]
\SetNoFillComment
\caption{$\alphaMech$}
\label{mech:sc-alpha}
{\bf Input:} Reported positions of agents that are affected by both facilities\;
{\bf Output:} Facility locations $\bw = (w_1,w_2)$\;
$i \gets \alpha n$-leftmost agent\;
$j \gets (1-\alpha)n$-leftmost agent\;
\tcp*[h]{{\bf (case 1.1)} $i$ and $j$ agree that $L$ is the most distant location} \\
\uIf{ $t(i) = t(j) = L$}{
    $w_1 \gets L$\;
    $w_2 \gets s(i)$\;
}
\tcp*[h]{{\bf (case 1.2)} $i$ and $j$ agree that $R$ is the most distant location} \\
\uElseIf{ $t(i) = t(j) = R$ }{
    $w_1 \gets R$\;
    $w_2 \gets s(j)$\;
}
\tcp*[h]{{\bf (case 2)} $i$ prefers $R$ and $j$ prefers $L$} \\
\Else{
    $w_1 \gets L$\;
    $w_2 \gets R$\;
}
\end{algorithm}

\begin{theorem}\label{thm:alpha-sp}
When all agents are affected by both facilities, $\alphaMech$ is strategyproof.
\end{theorem}

\begin{proof}
We first consider the case where agents $i$ and $j$ agree on which among $L$ and $R$ is the most distant location; without loss of generality, let this be $L$ (case 1.1). Recall that $\ell$ and $r$ are the second leftmost and second rightmost candidate locations, and it might be the case that $\ell = r$ as well as the case that neither of those locations exist. By definition, the solution returned by the mechanism is either $(L,R)$ in case $s(i)=R$ or $(L,\ell)$ in case $s(i) = \ell$. We switch between these two cases:
\begin{itemize}
\item 
The solution is $(L, R)$. Then, an agent would want to deviate only in case she prefers the solution $(L,\ell)$ or the solution $(r,R)$ over $(L,R)$.
If $x_i \geq R$, then, since $s(i)=R$, either $\ell$ and $r$ do not exist, or $\ell = r = R$; in both these cases, clearly no agent has an incentive to change the outcome. So, let us focus on the case where $x_i < R$.   
Agents that prefer $(L,\ell)$ over $(L, R)$ are necessarily located to the right of $i$; any misreport by such an agent can only move the location of the $\alpha n$-th agent to the left and $R$ would still be more preferred for this agent over $\ell$, thus not leading to the solution $(L,\ell)$ being chosen. Similarly, agents that prefer $(r,R)$ over $(L, R)$ are necessarily to the left of $i$; any misreport by such an agent can only move the $\alpha n$-th agent to the right and $L$ would still be the most-preferred location of this agent. 

\item 
The solution is $(L,\ell)$. Since $s(i)=\ell$, it has to be the case that $i$ is positioned weakly to the right of the midpoint between $\ell$ and $R$. Hence, $i$ and all agents to her right have no incentive to misreport as $(L,\ell)$ is the best outcome for them. If an agent to the left of $i$ misreports, then this can only lead to $i$ (and perhaps $j$ as well) being moved to the right, and thus the solution would remain the same since  it would still be the case that $i$ and $j$ agree that $L$ is the most distant location and $s(i) = \ell$.
\end{itemize}

The final case (case 2), where $i$ prefers $R$ and $j$ prefers $L$, leads to the solution $(L, R)$. 
If $i$ or any agent weakly to the left of $i$ misreport, then the solution either does not change or changes to $(L,\ell)$ which is a worse outcome for these agents as $R$ is their most-preferred location. 
Similarly, if agent $j$ and any agent weakly to the right of $j$ misreport, then the solution either does not change or changes to $(r, R)$ which is a worse outcome for them. 
Finally, observe that any agent strictly between $i$ and $j$ cannot change the outcome as any misreport would either push $i$ to the left (who would still prefer $R)$ or $j$ to the right (who would still prefer $L$).
\end{proof}

We next show an upper bound on the approximation ratio of $\alphaMech$ as a function of $\alpha$, which we will optimize later. 

\begin{theorem}\label{thm:doubleton:upper:deterministic}
When all agents are affected by both facilities, the approximation ratio of $\alphaMech$ is at most $\max\{2-\alpha, (1+\alpha)/(1-\alpha)\}$.
\end{theorem}

\begin{proof}
We begin by arguing that, in any worst-case instance, there are no agents at the left of $L$ or at the right of $R$. Assume otherwise and consider such an instance. We modify it by moving all agents at the left of $L$ to $L$, and all agents at the right of $R$ to $R$. Note that the output of the mechanism does not change, as the preference ordering of these agents for the available locations remains the same. Furthermore, if either $i$ (the $\alpha n$-th leftmost agent) or $j$ (the $(1-\alpha)n$-th leftmost agent) is among the relocated agents, her preference ordering does not change as well. Observe also that the social welfare of the optimal solution and of the mechanism reduce by the same amount in the modified instance; hence, the approximation ratio cannot decrease by such a move. 

Now let $\bw$ be the solution chosen by the mechanism, and let $\bo$ be an optimal solution. 
We consider each case of the mechanism separately, starting with the case where $i$ and $j$ agree on which among $L$ and $R$ is most distant; without loss of generality, let $L$ be that location (case 1.1). If the mechanism chooses the solution $\bw=(L,R)$, since all agents are in the interval $[L,R]$ as argued above, we have that $\SW(\bw) = n\cdot d(L,R)$. It suffices to consider the cases where the optimal solution $\bo$ is either $(L, \ell)$ or $(r, R)$; otherwise, if the optimal solution is $(L,R)$, the approximation ratio would be $1$. Let $S$ be the set of the $\alpha n$ leftmost agents, and thus $N\setminus S$ is the set of the remaining $(1-\alpha)n$ agents.
\begin{itemize}
    \item Case (A): $\bo = (L,\ell)$.  Since all agents are in the interval $[L,R]$ and $i$ prefers $R$ over $\ell$, the same is true for any agent $k \in S$; so, $d(k,\ell) \leq d(k,R)$ for any $k \in S$. In addition, since $\ell \in [L,R]$, $d(\ell,R) \leq d(L,R)$. Using these and the triangle inequality, we can bound the optimal social welfare as follows:
    \begin{align*}
        \SW(\bo) 
        &= \sum_k \bigg( d(k,L) + d(k,\ell) \bigg) \\
        &= \sum_{k \in S} \bigg( d(k,L) + d(k,\ell) \bigg) + \sum_{k \in N\setminus S} \bigg( d(k,L) + d(k,\ell) \bigg) \\
        &\leq \sum_{k \in S} \bigg( d(k,L) + d(k,R) \bigg) + \sum_{k \in N \setminus S} \bigg( d(k,L) + d(k,R) + d(\ell,R) \bigg) \\
        &= |S| \cdot d(L,R) + |N \setminus S| \cdot 2 d(L,R) \\
        &= (2-\alpha)n \cdot d(L,R). 
    \end{align*}
    Hence, the approximation ratio is at most $2-\alpha$. 
    
    \item Case (B): $\bo = (r,R)$. Since all agents are in the interval $[L,R]$ and $i$ prefers $L$ over $R$, the same is true for any agent $k \in N \setminus S$; so, $d(k,R) \leq d(k,L)$ for any $k \in N\setminus S$, which also implies that $d(k,r) \leq d(k,L)$ for any $k \in N\setminus S$ since $r \in [L,R]$. In addition, $d(L,r) \leq d(L,R)$. Using these and the triangle inequality, we can bound the optimal social welfare as follows:
    \begin{align*}
        \SW(\bo) 
        &= \sum_k \bigg( d(k,r) + d(k,R) \bigg) \\
        &= \sum_{k \in S} \bigg( d(k,r) + d(k,R) \bigg) + \sum_{k \in N\setminus S} \bigg( d(k,r) + d(k,R) \bigg) \\
        &\leq \sum_{k \in S} \bigg( d(k,L) + d(L,r) + d(k,R) \bigg) + \sum_{k \in N\setminus S} \bigg( d(k,L) + d(k,R) \bigg) \\
        &\leq |S| \cdot 2 d(L,R) + |N\setminus S| \cdot d(L,R) \\
        &= (1+\alpha) n \cdot d(L,R).
    \end{align*}
    Hence, the approximation ratio is at most $1+\alpha \leq 2 -\alpha$ since $\alpha \leq 1/2$.    
\end{itemize}
If the mechanism returns the solution $\bw = (L,\ell)$ when $i$ and $j$ both prefer $L$ over $R$, then all $(1-\alpha)n$ agents of $N\setminus S$ (who are at the right of $i$) must agree with the preference of $i$. So, any agent $k \in N\setminus S$ prefers $\ell$ over $R$, and has utility $d(k,L) + d(k,\ell) \geq d(k,L) + d(k,R) = d(L,R)$, leading to a social welfare of $\SW(\bw) \geq (1-\alpha)n \cdot d(L,R)$. We now consider the possible structures of the optimal solution $\bo$. First, if $\bo = (L,\ell)$, then clearly the approximation ratio is $1$. If $\bo = (L,R)$, then, since all agents are in the interval $[L,R]$, $\SW(\bo) = n \cdot d(L,R)$, and the approximation ratio is at most $1/(1-\alpha)$. Finally, if $\bo = (r,R)$, then using the same arguments in case (B) above, we can show that $\SW(\bo) \geq (1+\alpha) n \cdot d(L,R)$, and the approximation ratio is at most $(1+\alpha)/(1-\alpha)$.

We now switch to the case where $i$ and $j$ disagree on their preference between $L$ and $R$ (case 2). Then, the mechanism returns  the solution $\bw = (L,R)$ with social welfare $\SW(\bw) = n \cdot d(L,R)$. It suffices to consider the case where the optimal solution $\bo$ is either $(L, \ell)$ or $(r, R)$; without loss of generality, suppose that $\bo = (L, \ell)$. Since all agents are in the interval $[L,R]$ and $i$ prefers $R$ over $L$, the same is true for any of the $\alpha n$ leftmost agents. So, $d(k,L) \leq d(k,R)$ for $k \in S$, which also implies that $d(k,\ell) \leq d(k,R)$ for any $k \in S$ since $\ell \in [L,R]$. In addition, $d(\ell,R) \leq d(L,R)$. 
Using these and the triangle inequality, we can now bound the optimal social welfare as follows:
\begin{align*}
    \SW(\bo) &= \sum_k \bigg( d(k,\ell) + d(k,L) \bigg) \\
    &= \sum_{k \in S} \bigg( d(k,\ell) + d(k,L) \bigg) + \sum_{k \in N\setminus S} \bigg( d(k,\ell) + d(k,L) \bigg) \\
    &\leq \sum_{k \in S} \bigg( d(k,R) + d(k,L) \bigg) + \sum_{k \in N\setminus S} \bigg( d(k,R) + d(\ell,R) + d(k,L) \bigg) \\
    &\leq |S| \cdot d(L,R) + |N\setminus S| \cdot 2 d(L,R) \\
    &= (2-\alpha)n \cdot d(L,R).
\end{align*}
So, the approximation ratio in this case is at most $2-\alpha$ as well.

We conclude that the approximation ratio is bound by above from $\max\{2-\alpha, (1+\alpha)/(1-\alpha)\}$.
\end{proof}

By setting $\alpha = 2-\sqrt{3}$ we obtain the following upper bound.

\begin{theorem}
The approximation ratio of {\text{\sc $(2-\sqrt{3})$-Statistic}} is at most $\sqrt{3}$.    
\end{theorem}

Next, we present a tight lower bound of $\sqrt{3}$ on the approximation ratio of all deterministic strategyproof mechanisms, thus completely resolving this case. 

\begin{theorem}\label{thm:doubleton:lower:deterministic}
When all agents are affected by both facilities, the approximation ratio of any deterministic strategyproof mechanism is at least $\sqrt{3}$. 
\end{theorem}

\begin{proof}
We consider instances with candidate locations $\{0,0,2,2\}$. Let $\alpha = \sqrt{3}-1$. 
Our first instance $I$ is such that there are $\alpha n$ agents at $0$ and $(1-\alpha)n$ agents at $1+\varepsilon$, where $\varepsilon >0$ is an infinitesimal\footnote{For simplicity, we assume that $n$ is large enough so that the cardinalities of the agents at $0$ and $1+\varepsilon$ are integers.}; See Figure~\ref{fig:doubleton:lower:deterministic:I}.
Clearly,
$\SW(0,0) \approx 2(1-\alpha) n$, 
$\SW(0,2) \approx 2\alpha n + 2(1-\alpha)n = 2n$,
and $\SW(2,2) \approx 4 \alpha n + 2(1-\alpha)n = 2(1+\alpha)n$.
Hence, if the mechanism chooses the solution $(0,0)$, the approximation ratio is at least 
\begin{align*}
    \frac{1 + \alpha}{1-\alpha} = \frac{\sqrt{3}}{2-\sqrt{3}} = 3 + 2\sqrt{3},
\end{align*}
whereas if it chooses the solution $(0,2)$, the approximation ratio is at least $1+\alpha = \sqrt{3}$. Hence, the mechanism must choose the optimal solution $(2,2)$ when given $I$ as input. 

We now consider a sequence of instances starting with $I$ such that the agents positioned at $1+\varepsilon$ move one by one to position $2$. Let $I_t$ be the instance in this sequence where $t$ agents have moved from $1+\varepsilon$ to $2$ and note that $I_0 = I$. For each instance $I_{t+1}$, the solution chosen by the mechanism must still be $(2,2)$ as otherwise, if some facility is placed at $0$ in $I_{t+1}$, the agent that moves from $1+\varepsilon$ in $I_t$ to $2$ in $I_{t+1}$ would misreport its position at $2$ when the instance is $I_t$, thus leading to $I_{t+1}$, to increase its utility from $2-2\varepsilon$ to at least $2$. 

Let $J$ be the last instance of the previous sequence in which there are $\alpha n$ agents at $0$ and $(1-\alpha)n$ agents at $2$; see Figure~\ref{fig:doubleton:lower:deterministic:J}. We next consider another sequence of instances starting with $J$ such that the agents positioned at $0$ move one by one to position $1-\varepsilon$. Let $J_t$ be the instance in this sequence where $t$ agents have moved from $0$ to $1-\varepsilon$ and note that $J_0 = J$. For each instance $J_{t+1}$ in this sequence, the solution chosen by the mechanism must still be $(2,2)$ as otherwise, if some facility is placed at $0$ in $J_{t+1}$, the agent that moves from $0$ in $J_t$ to $1-\varepsilon$ in $J_{t+1}$ would misreport its position as $0$, leading back to $J_t$, to increase its utility from at most $2$ to $2+2\varepsilon$. 

Let $Q$ be the last instance of this sequence in which there are $\alpha n$ agents at $1-\varepsilon$, $(1-\alpha)n$ agents at $2$, and the solution chosen by the mechanism is $(2,2)$; see Figure~\ref{fig:doubleton:lower:deterministic:Q}. Since $\SW(2,2) \approx 2\alpha n$ and $\SW(0,0) = 2\alpha n + 4(1-\alpha) n = 2(2-\alpha)n$, the approximation ratio is at least 
\begin{align*}
    \frac{2-\alpha}{\alpha} = \frac{2}{\sqrt{3}-1} - 1 = \sqrt{3}.
\end{align*}
This completes the proof.
\end{proof}

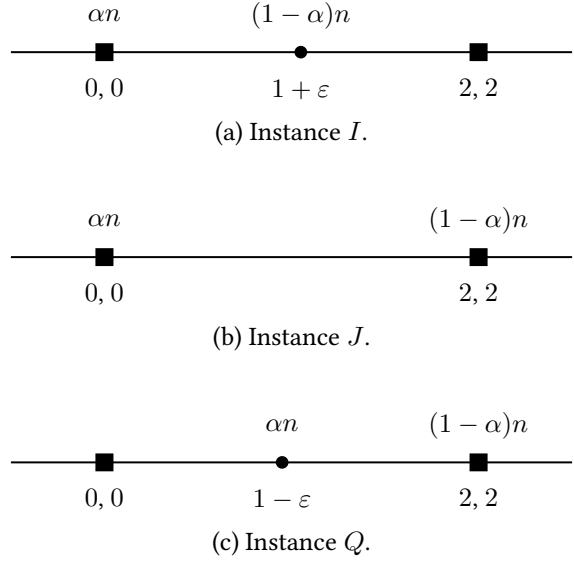
\begin{figure}[t]
\tikzset{every picture/.style={line width=0.75pt}} 
\centering
\begin{subfigure}[t]{\linewidth}
\centering
\begin{tikzpicture}[x=0.7pt,y=0.7pt,yscale=-1,xscale=1]
\draw [line width=0.75]  (0,0) -- (300,0) ;
\filldraw ([xshift=-3pt,yshift=-3pt]50,0) rectangle ++(6pt,6pt);
\filldraw ([xshift=-3pt,yshift=-3pt]250,0) rectangle ++(6pt,6pt);
\filldraw (155,0) circle (2pt);

\draw (50,-20) node [inner sep=0.75pt]  [font=\small]  {$\alpha n$};
\draw (155,-20) node [inner sep=0.75pt]  [font=\small]  {$(1-\alpha) n$};

\draw (50,20) node [inner sep=0.75pt]  [font=\small]  {$0$, $0$};
\draw (250,20) node [inner sep=0.75pt]  [font=\small]  {$2$, $2$};

\draw (155,20) node [inner sep=0.75pt]  [font=\small]  {$1+\varepsilon$};
\end{tikzpicture}
\caption{Instance $I$.}
\label{fig:doubleton:lower:deterministic:I}
\end{subfigure}
\\[20pt]
\begin{subfigure}[t]{\linewidth}
\centering
\begin{tikzpicture}[x=0.7pt,y=0.7pt,yscale=-1,xscale=1]
\draw [line width=0.75]  (0,0) -- (300,0) ;
\filldraw ([xshift=-3pt,yshift=-3pt]50,0) rectangle ++(6pt,6pt);
\filldraw ([xshift=-3pt,yshift=-3pt]250,0) rectangle ++(6pt,6pt);

\draw (50,-20) node [inner sep=0.75pt]  [font=\small]  {$\alpha n$};
\draw (250,-20) node [inner sep=0.75pt]  [font=\small]  {$(1-\alpha) n$};

\draw (50,20) node [inner sep=0.75pt]  [font=\small]  {$0$, $0$};
\draw (250,20) node [inner sep=0.75pt]  [font=\small]  {$2$, $2$};

\end{tikzpicture}
\caption{Instance $J$.}
\label{fig:doubleton:lower:deterministic:J}
\end{subfigure}
\\[20pt]
\begin{subfigure}[t]{\linewidth}
\centering
\begin{tikzpicture}[x=0.7pt,y=0.7pt,yscale=-1,xscale=1]
\draw [line width=0.75]  (0,0) -- (300,0) ;
\filldraw ([xshift=-3pt,yshift=-3pt]50,0) rectangle ++(6pt,6pt);
\filldraw ([xshift=-3pt,yshift=-3pt]250,0) rectangle ++(6pt,6pt);
\filldraw (145,0) circle (2pt);

\draw (145,-20) node [inner sep=0.75pt]  [font=\small]  {$\alpha n$};
\draw (250,-20) node [inner sep=0.75pt]  [font=\small]  {$(1-\alpha) n$};

\draw (50,20) node [inner sep=0.75pt]  [font=\small]  {$0$, $0$};
\draw (250,20) node [inner sep=0.75pt]  [font=\small]  {$2$, $2$};

\draw (145,20) node [inner sep=0.75pt]  [font=\small]  {$1-\varepsilon$};

\end{tikzpicture}
\caption{Instance $Q$.}
\label{fig:doubleton:lower:deterministic:Q}
\end{subfigure}
\caption{The instances considered in the proof of Theorem~\ref{thm:doubleton:lower:deterministic} when the solution chosen for instance $I$ is $(2,2)$. Instance $J$ is obtained by moving the $(1-\alpha)n$ agents at $1+\varepsilon$ in $I$ to $2$. Instance $Q$ is obtained by moving the $\alpha n$ agents at $0$ in $J$ to $1-\varepsilon$. In all these instances, the mechanism must choose solution $(2,2)$ due to strategyproofness, which leads to the lower bound of $\sqrt{3}$ on its approximation ratio.}
\label{fig:doubleton:lower:deterministic}
\end{figure}

\subsection{Randomized Mechanisms} \label{sec:doubleton:randomized}
We now focus on randomized strategyproof mechanisms for instances in which all agents are affected by both facilities and first show that an improved upper bound of approximately $1.522$ can be achieved by the following $\UniformAlpha$ mechanism:
Pick an integer number $k$ uniformly at random in the interval $[1, n/2]$, and run 
$\alphaMech$ for $\alpha = k/n$. This mechanism is clearly strategyproof as it is a (uniform) probability distribution over strategyproof mechanisms. We next focus on bounding its approximation ratio. 

\begin{theorem} \label{thm:doubleton:randomized:1.52}
When all agents are affected by both facilities, the approximation ratio of \hspace{0.075cm}$\UniformAlpha$ is at most $(5+4\sqrt{2})/7 \approx 1.522$, and this is tight.  
\end{theorem}

\begin{proof}
By applying the same reasoning as in the proof of Theorem~\ref{thm:doubleton:upper:deterministic}, we can assume that all agents are in the interval $[L,R]$. Without loss of generality, we further assume that the majority of the agents prefer $R$ over $L$. We now partition the agents into the following sets:
\begin{itemize}
    \item $S_<$ consists of any agent $i$ such that $x_i \leq (L+r)/2$, and thus $t(i) = R$ and $s(i)=r$;
    \item $S$ consists of any agent $i$ such that $(L+r)/2 < x_i \leq (L+R)/2$, and thus $t(i) = R$ and $s(i)=L$;
    \item $S_>$ consists of any agent $i$ such that $x_i > (L+R)/2$, and thus $t(i)=L$ (note that $s(i)$ might be any of $R$ or $\ell$; this makes no difference in the behavior of the mechanism as explained below). 
\end{itemize}
Let $a = |S_<|$, $b = |S|$, and $c = |S_>|$. Observe that $a+b+c = n$ and, by our assumption that most agents prefer $R$ over $L$, $a+b \geq n/2$. Let us now argue a bit about the outcome of $\alphaMech$ depending on the positions of the $\alpha n$-leftmost agent $i$ and the $(1-\alpha)n$-th leftmost agent $j$. Clearly, $i \not\in S_>$. 
\begin{itemize}
    \item If $i \in S_<$ and $j \in S_<$, then, since $t(i) = t(j) = R$ and $s(j)=r$, $\alphaMech$ returns $(r,R)$;
    \item If $i \in S_<$ and $j \in S \cup S_>$, then, since $t(i) = R$ and either $t(j) = R$ and $s(j) = L$ (when $j \in S$) or $t(j) = L$ (when $j \in S_>$),  $\alphaMech$ returns $(L,R)$;
    \item Similarly, if $i \in S$ and $j \in S \cup S_>$, then, again $t(i) = R$ and either $t(j) = R$ and $s(j) = L$, or $t(j) = L$, when $\alphaMech$ returns $(L,R)$.
\end{itemize}
Hence, the solution is $(L,R)$ in any case other than when $i, j \in S_<$. Since this is only possible when $a > n/2$, our mechanism works as follows: If $a > n/2$, it returns $(L,R)$ with probability $\frac{b+c}{n/2} = \frac{n-a}{n/2}$ and $(r,R)$ otherwise. If $a \leq n/2$, it always returns $(L,R)$. Next, we consider each case separately and argue about the approximation ratio of the mechanism. 

\medskip

\noindent
{\bf Case 1: $a \leq n/2$.} Since the mechanism returns the solution $(L,R)$ with probability $1$, and all agents are in the interval $[L,R]$, we have that 
$\mathbb{E}[\SW(\bw)] = n \cdot d(L,R)$. Since the approximation ratio is trivially $1$ if the optimal solution is $(L,R)$, it suffices to consider the two remaining potential optimal solutions, $(r,R)$ and $(L,\ell)$. 
For the solution $(r,R)$, by definition, we have
\begin{align*}
    \SW(r,R) = \sum_{i \in S_< \cup S \cup S_>} \bigg( d(i,r) + d(i,R) \bigg).
\end{align*}
We make the following observations about the agents in the different sets:
\begin{itemize}
    \item For any agent $i \in S_<$, since $x_i \in [L, R]$, 
    $$d(i,r) + d(i,R) \leq 2 \cdot d(L,R).$$ 
    \item For any agent $i \in S$, since $x_i \geq (L+r)/2$, 
    $$d(i,r) + d(i,R) \leq d((L+r)/2,r) + d((L+r)/2,R) = d(L,R).$$
    \item For any agent $i \in S_>$, for whom $x_i > (L+R)/2$, we consider two cases:
    \begin{itemize}
        \item $r \leq (L+R)/2$. Then, $d(i,r) + d(i,R) = d(r,R) \leq d(L,R)$.
        \item $r > (L+R)/2$. Then,  $d(i,r) + d(i,R) \leq 2 d((L+R)/2,R) = d(L,R)$.
    \end{itemize}
\end{itemize}
Using these observations along with the facts that $a \leq n/2$ and $a+b+c=n$, we obtain
\begin{align*}
    \SW(r,R) \leq (2a + b + c) \cdot d(L,R) \leq \frac{3n}{2} \cdot d(L,R) \leq \frac{3}{2} \cdot \mathbb{E}[\SW(\bw)]. 
\end{align*}
For the solution $(L,\ell)$, by the definition of the sets $S_<$, $S$ and $S_>$, we have 
\begin{align*}
    \SW(L,\ell) &= \sum_{i \in S_< \cup S \cup S_>} \bigg( d(i,L) + d(i,\ell) \bigg) \\
    &\leq \sum_{i \in S_< \cup S} \bigg( d(i,L) +d(i,R) \bigg) + 2 \cdot \sum_{i \in S_>} d(i,L) \\
    &\leq (a+b + 2c) \cdot d(L,R).
\end{align*}
Since $a+b+c = n$ and $c \leq n/2$ (due to the fact that $a+b \geq n/2$), we finally obtain
\begin{align*}
    \SW(L,\ell) \leq \frac{3n}{2} \cdot d(L,R) \leq \frac{3}{2} \cdot \mathbb{E}[\SW(\bw)].
\end{align*}
Consequently, no matter which solution is the optimal one, the approximation ratio is at most $3/2$ in the case where $a \leq n/2$.

\medskip

\noindent 
{\bf Case 2: $a > n/2$.} Recall that in this case, the mechanism returns the solution $(L,R)$ with probability $\frac{n-a}{n/2} = \frac{2(n-a)}{n}$ and the solution $(r,R)$ with the remaining probability $\frac{2a-n}{n}$. We consider each possible optimal solution separately. 
\begin{itemize}
    \item The optimal solution is $(L,R)$. Since all agents are in the interval $[L,R]$, we have that $\SW(L,R) = n \cdot d(L,R)$.
    Since 
    \begin{align*}
    \SW(r,R) \geq a \bigg( d\left(\frac{L+r}{2},r\right) + d\left(\frac{L+r}{2},R\right) \bigg) = a \cdot d(L,R),
    \end{align*}
    the expected social welfare of the solution computed by the mechanism is
    \begin{align*}
    \mathbb{E}[\SW(\bw)] 
    &= \frac{2(n-a)}{n}\cdot \SW(L,R) +\frac{2a-n}{n} \cdot \SW(r,R) \\
    &\geq \frac{2(n-a)}{n} \cdot n \cdot d(L,R) +\frac{2a-n}{n} \cdot a \cdot d(L,R) \\
    &= \frac{2n^2 - 3an + 2a^2}{n} \cdot d(L,R) \\
    &\geq \frac{7}{8}n \cdot d(L,R),
    \end{align*}
    where the last inequality follows by the fact that the function $2n^2 - 3an + 2a^2$ is minimized for $a = 3n/4$. Hence, the approximation ratio is at most $8/7$ in this case. 
 
    \item The optimal solution is $(r,R)$. As already shown above, we have that 
    $$\SW(r,R) \leq (2a+b+c) \cdot d(L,R) = (n+a) \cdot d(L,R).$$
    Since the expected social welfare is 
    \begin{align*}
    \mathbb{E}[\SW(\bw)] 
    &= \frac{2(n-a)}{n}\cdot \SW(L,R) +\frac{2a-n}{n} \cdot \SW(r,R), 
    \end{align*}
    the approximation ratio $\frac{\SW(r,R)}{\mathbb{E}[\SW(\bw)]}$ is increasing in terms of $\SW(r,R)$, and hence we can upper-bound it as follows:
    \begin{align*}
        \frac{\SW(r,R)}{\mathbb{E}[\SW(\bw)]} \leq \frac{n+a}{\frac{2(n-a)}{n}\cdot n + \frac{2a-n}{n} \cdot (n+a)} 
        = \frac{n^2 + n a}{n^2 - na + 2a^2}.
    \end{align*}
    The last expression is a decreasing function in terms of $a$ for any $a > n/2$, and thus the approximation ratio is at most $3/2$. 

    \item The optimal solution is $(L,\ell)$. As already shown above, we have that
    $$\SW(L,\ell) \leq (a + b + 2c) \cdot d(L,R) \leq (2n-a) \cdot d(L,R).$$
    and 
    $$\mathbb{E}[\SW(\bw)] \geq \frac{2n^2 - 3an + 2a^2}{n} \cdot d(L,R).$$
    Hence, the approximation ratio is 
    \begin{align*}
        \frac{\SW(L,\ell)}{\mathbb{E}[\SW(\bw)]} \leq \frac{2n^2 - an }{2n^2 - 3an + 2a^2}.
    \end{align*}
    The last expression is maximized to $\frac{5+4\sqrt{2}}{7} \approx 1.522$ for $a = \frac{2+3\sqrt{2}}{5+4\sqrt{2}}$, leading to the desired upper bound on the approximation ratio.
\end{itemize}
The proof of the upper bound is now complete. 

We finally show that the analysis is tight with the following instance:
\begin{itemize}
    \item There are two candidate locations at $0$ and two candidate locations at $2$.
    \item There are $(2-\sqrt{2})n$ agents at $1-\varepsilon$, for some infinitesimal $\varepsilon > 0$, and the remaining $(\sqrt{2}-1)n$ agents are at $2$.\footnote{We use irrational numbers for clarity of exposition; technically, we would need to use floors or ceilings, with an imperceptible loss in the approximation ratio for large enough $n$.}
\end{itemize}
We have that $\SW(0,0) \approx 2\sqrt{2} n$, $\SW(0,2) \approx 2 n$, and $\SW(2,2) \approx 2(2-\sqrt{2}) n$. So, the optimal solution is $(0,0)$. Our mechanism returns the solution $(0,2)$ with probability $\frac{(\sqrt{2}-1)n}{n/2} = 2(\sqrt{2}-1)$ and the solution $(2,2)$ with the remaining probability $3-2\sqrt{2}$. Hence, its expected social welfare is
$$2n\cdot2(\sqrt{2}-1) + 2(2-\sqrt{2})n\cdot(3-2\sqrt{2}) = (16-10\sqrt{2})n,$$
leading to an approximation ratio of $(5+4\sqrt{2})/7$.
\end{proof}

We complement the above positive result with an almost tight lower bound of $3/2$ on the approximation ratio of any randomized mechanism. 

\begin{theorem}\label{thm:doubleton:lower:randomized}
When all agents are affected by both facilities, the approximation ratio of any randomized strategyproof mechanism is at least $3/2$.
\end{theorem}

\begin{proof}
We consider instances with candidate locations $\{0,0,2,2\}$. Our first instance $I$ is such that there are two agents, one at $1-\varepsilon$ and one at $1+\varepsilon$; see Figure~\ref{fig:doubleton:lower:randomized:I}. Let $p$ be the probability with which a randomized strategyproof mechanism $M$ returns the solution $(0,0)$ and $q$ be the probability that $M$ returns $(2,2)$; hence the probability of the last possible solution $(0,2)$ is $1-p-q$ . 
Without loss of generality, suppose that $p \geq q$. 
The expected utility of agent $i$ who is positioned at $1-\varepsilon$ in $I$ is $2-2\varepsilon(p-q)$.

Now consider an instance $J$ that differs from $I$ only in that agent $i$ is now positioned at $0$ instead of $1-\varepsilon$; see Figure~\ref{fig:doubleton:lower:randomized:J}. 
Let $p', q'$ be the probabilities that outcomes $(0,0)$ and $(2,2)$ are returned by mechanism $M$ when given $J$ as input. 
Note that the optimal solution for $J$ is to place both facilities at $2$, for a social welfare of $6-2\varepsilon$, 
while the randomized solution computed by $M$ has expected social welfare $4-(2-2\varepsilon)(p' - q')$.

Since $M$ is strategyproof, agent $i$ should not have incentive to misreport her true position of $1-\varepsilon$ in $I$ as $0$, thus leading to $J$. This happens when $2-2\varepsilon (p-q) \geq 2-2\varepsilon (p'-q')$, which implies that $p - q \leq p' - q'$. Using this, as well as the assumption that $p \geq q$, we have that the expected social welfare of the mechanism when given $J$ as input is 
$4-(2-2\varepsilon)(p - q) \leq 4$, and thus the approximation ratio is at least $3/2$.
\end{proof}

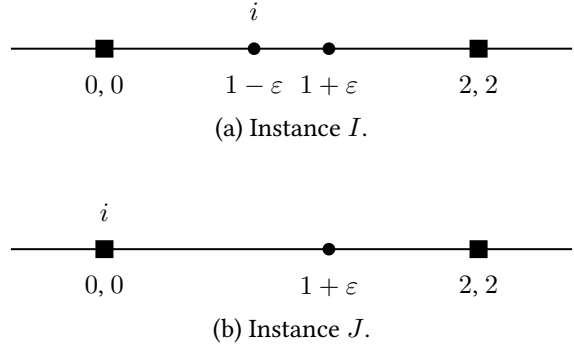
\begin{figure}[t]
\tikzset{every picture/.style={line width=0.75pt}} 
\centering
\begin{subfigure}[t]{\linewidth}
\centering
\begin{tikzpicture}[x=0.7pt,y=0.7pt,yscale=-1,xscale=1]
\draw [line width=0.75]  (0,0) -- (300,0) ;
\filldraw ([xshift=-3pt,yshift=-3pt]50,0) rectangle ++(6pt,6pt);
\filldraw ([xshift=-3pt,yshift=-3pt]250,0) rectangle ++(6pt,6pt);
\filldraw (130,0) circle (2pt);
\filldraw (170,0) circle (2pt);

\draw (130,-20) node [inner sep=0.75pt]  [font=\small]  {$i$};

\draw (50,20) node [inner sep=0.75pt]  [font=\small]  {$0$, $0$};
\draw (250,20) node [inner sep=0.75pt]  [font=\small]  {$2$, $2$};

\draw (130,20) node [inner sep=0.75pt]  [font=\small]  {$1-\varepsilon$};
\draw (170,20) node [inner sep=0.75pt]  [font=\small]  {$1+\varepsilon$};
\end{tikzpicture}
\caption{Instance $I$.}
\label{fig:doubleton:lower:randomized:I}
\end{subfigure}
\\[20pt]
\begin{subfigure}[t]{\linewidth}
\centering
\begin{tikzpicture}[x=0.7pt,y=0.7pt,yscale=-1,xscale=1]
\draw [line width=0.75]  (0,0) -- (300,0) ;
\filldraw ([xshift=-3pt,yshift=-3pt]50,0) rectangle ++(6pt,6pt);
\filldraw ([xshift=-3pt,yshift=-3pt]250,0) rectangle ++(6pt,6pt);

\filldraw (170,0) circle (2pt);

\draw (50,-20) node [inner sep=0.75pt]  [font=\small]  {$i$};

\draw (50,20) node [inner sep=0.75pt]  [font=\small]  {$0$, $0$};
\draw (250,20) node [inner sep=0.75pt]  [font=\small]  {$2$, $2$};

\draw (170,20) node [inner sep=0.75pt]  [font=\small]  {$1+\varepsilon$};
\end{tikzpicture}
\caption{Instance $J$.}
\label{fig:doubleton:lower:randomized:J}
\end{subfigure}
\caption{The instances considered in the proof of Theorem~\ref{thm:doubleton:lower:randomized}.}
\label{fig:doubleton:lower:randomized}
\end{figure}

\section{General Optional Preferences}\label{sec:general}
In this section we turn our attention to general instances which might consist of some agents that are affected by a single facility and some agents that are affected by both. We first show a tight bound of $3$ on the approximation ratio of deterministic strategyproof mechanisms, and then show that the approximation ratio of randomized mechanisms is between $3/2$ and $2$. 

\subsection{Deterministic Mechanisms} \label{sec:general:deterministic}
We first show the upper bound by considering the {\sc LR-Stronger-Majority} mechanism, which works as follows. 
For any $j \in [2]$, the mechanism defines $S_j$ to be the set consisting of the {\em majority} of agents in $N_j$ that agree on their preference over $L$ and $R$. Let $j \in [2]$ be such that the agents of $S_j$ have a stronger majority compared to the agents of $S_{3-j}$, that is, it holds that 
$$|S_j| - \frac{n_j}{2} \geq |S_{3-j}| - \frac{n_{3-j}}{2} 
\Leftrightarrow 2|S_j| - n_j \geq 2|S_{3-j}| - n_{3-j}.$$
Then, $w_j$ is set to be whichever between $L$ and $R$ the agents of $S_j$ prefer, and $w_{3-j}$ is the remaining location between $L$ and $R$.  See Mechanism~\ref{mech:LR-stronger}.

\SetCommentSty{mycommfont}
\begin{algorithm}[h]
\SetNoFillComment
\caption{\sc LR-Stronger-Majority}
\label{mech:LR-stronger}
{\bf Input:} Reported positions of agents\;
{\bf Output:} Facility locations $\bw = (w_1,w_2)$\;
\For{$j \in [2]$}{
    $S_j(L) \gets \{i \in N_j: d(i,L) \geq d(i,R)\}$\;
    $S_j(R) \gets N_j \setminus S_j(L)$\;
    \uIf{$|S_j(L)| \geq |S_j(R)|$}{
        $S_j \gets S_j(L)$\;
        $P_j \gets L$ \;
    }
    \Else{
        $S_j \gets S_j(R)$\;
        $P_j \gets R$ \;
    }
}
\uIf{$2|S_1| - n_1 \geq 2|S_2| - n_2$}{
    $w_1 \gets P_1$\;
    $w_2 \gets \{L,R\}\setminus P_1$\;
}
\Else{
    $w_2 \gets P_2$\;
    $w_1 \gets \{L,R\}\setminus P_2$\;
}

\end{algorithm}

\begin{theorem}\label{thm:general:upper:sp}
{\sc LR-Stronger-Majority} is strategyproof.
\end{theorem}

\begin{proof}
First observe that any agent that is affected by both facilities has no incentive to deviate since the two facilities are always placed at $L$ and $R$. Now, without loss of generality, suppose that $2|S_1| - n_1 \geq 2|S_2| - n_2$ and that the solution computed by the mechanism is $\bw = (L,R)$, that is, $F_1$ is placed at $L$ because the majority of the agents in $N_1$ prefer it over $R$, and $F_2$ is placed at $R$. We consider each agent depending on her preference.
\begin{itemize}
    \item Any agent that is affected only by $F_1$ that is part of $S_1$ has no incentive to affect the outcome since she prefers $L$ over $R$.
    \item Any agent that is affected only by $F_1$ that is not in $S_1$ cannot change the outcome since she can only misreport a position that would either not change or increase $|S_1|$.
    \item Any agent that is affected only by $F_2$ that is part of $S_2$ either prefers $R$ over $L$, or cannot change the outcome since she can only misreport a position that would either not change or decrease $|S_2|$.
    \item Any agent that is affected only by $F_2$ that is not part of $S_2$ either prefers $R$ over $L$, or she prefers $L$ over $R$ which means that the agents of $S_2$ prefer $R$ and increasing $|S_2|$ cannot change the outcome as $F_2$ would again be at $R$ if the agents of $N_2$ obtained the stronger majority. 
\end{itemize}
So, the mechanism is strategyproof. 
\end{proof}

\begin{theorem}\label{thm:general:upper:deterministic}
The approximation ratio of {\sc LR-Stronger-Majority} is at most $3$.
\end{theorem}

\begin{proof}
Suppose without loss of generality that $2|S_1| - n_1 \geq 2|S_2| - n_2$, and that the agents of $S_1$ prefer $L$ over $R$. Hence, the solution computed by the mechanism is $\bw = (L,R)$. Let $\bo = (o_1,o_2)$ be an optimal solution. Let $m_1$ and $m_2$ be the median agents of $N_1$ and $N_2$, and observe that $m_1 \in S_1$ and $m_2 \in S_2$ by definition. 
We consider the following two cases:

\medskip
\noindent
{\bf Case 1: $d(m_2,o_2) \leq d(m_2,R) = d(m_2,w_2)$.}
Since $m_1$ prefers $L$ over $R$, $w_1=L$ is the most distant location from $m_1$ out of all candidate locations, and thus $d(m_1,w_1) = d(m_1,L) \geq d(m_1,o_1)$. In addition, the median agent $m_j$ minimizes the total distance of the agents in $N_j$ for any $j \in [2]$. 
Using these in combination with the triangle inequality, we have
\begin{align*}
    \SW(\bo) = \sum_{j \in [2]} \sum_{i \in N_j} d(i,o_j) 
    &\leq \sum_{j \in [2]} \sum_{i \in N_j} \bigg( d(i,m_j) + d(m_j,o_j) \bigg) \\
    &\leq \sum_{j \in [2]} \sum_{i \in N_j} \bigg( d(i,m_j) + d(m_j,w_j) \bigg) \\
    &\leq \sum_{j \in [2]} \sum_{i \in N_j} \bigg( 2d(i,m_j) + d(i,w_j) \bigg) \\ 
    &\leq 3 \cdot \sum_{j \in [2]} \sum_{i \in N_j} d(i,w_j). \\
    &= 3 \cdot \SW(\bw).
\end{align*}

\medskip
\noindent
{\bf Case 2: $d(m_2,o_2) > d(m_2,R) = d(m_2,w_2)$.}
For this to be possible, it must be the case that $o_2 = L$. We first present a lower bound on the social welfare of $\bw$. By the definition of the sets $S_1$ and $S_2$, we have that $|S_1|$ agents of $N_1$ are weakly to the right of $(L+R)/2$ and $n_2 - |S_2|$ agents of $N_2$ are weakly to the left of $(L+R)/2$. Hence
\begin{align*}
&\SW(\bw) = \sum_{i \in N_1} d(i,L) + \sum_{i \in N_2} d(i,R) \geq \bigg( |S_1| + n_2 - |S_2| \bigg) \cdot \frac{d(L,R)}{2} \\
&\Leftrightarrow d(L,R) \leq \frac{2}{|S_1| + n_2 - |S_2|} \cdot \SW(\bw).
\end{align*}

We now upper bound the contribution of the agents to the optimal social welfare, starting with the agents of $N_2$. 
By the definition of $S_2$, and using the triangle inequality, we have
\begin{align*}
    \sum_{i \in N_2} d(i,o_2) = \sum_{i \in N_2} d(i,L) &= \sum_{i \in S_2} d(i,L) + \sum_{i \not\in S_2} d(i,L) \\
    &\leq \sum_{i \in S_2} \bigg( d(i,R) + d(L,R) \bigg) + \sum_{i \not\in S_2} d(i,R) \\
    &= \sum_{i \in N_2} d(i,w_2) + |S_2| \cdot d(L,R).
\end{align*}
For the agents of $N_1$, since $o_2=L$, the optimal location $o_1$ of $F_1$ is either the location $\ell$ directly to the right of $L$ (if it exists) or $R$. 
By definition (and our assumption that $w_1 = L$), for any agent $i \in S_1$, we have that $d(i,R) \leq d(i,L)$, which further implies that $d(i,\ell) \leq d(i,L)$ since $\ell \in [L,R]$. So, $d(i,o_1) \leq d(i,L)$ for every $i \in S_1$. 
Using this together with the triangle inequality for the agents not in $S_1$, the fact that $o_1 \in [L,R]$, and $w_1 = L$, we obtain
\begin{align*}
    \sum_{i \in N_1} d(i,o_1) &= \sum_{i \in S_1} d(i,o_1) + \sum_{i \not\in S_1} d(i,o_1) \\
    &\leq \sum_{i \in S_1} d(i,L) + \sum_{i \not\in S_1} \bigg( d(i,L) + d(L,o_1) \bigg) \\
    &\leq \sum_{i \in N_1} d(i,w_1) + (n_1 - |S_1|) \cdot d(L,R). 
\end{align*}
By putting everything together, we have
\begin{align*}
    \SW(\bo) &\leq \SW(\bw) + (n_1 - |S_1|) \cdot d(L,R) + |S_2| \cdot d(L,R) \\
    &\leq \bigg( 1 + 2\cdot \frac{n_1-|S_1|+|S_2|}{|S_1|+n_2-|S_2|} \bigg) \cdot \SW(\bw) \\
    &\leq 3 \cdot \SW(\bw),
\end{align*}
where the last inequality follows since $2|S_1| - n_1 \geq 2|S_2| -n_2 \Leftrightarrow n_1 \leq 2|S_1| - 2|S_2| + n_2$, and thus
\begin{align*}
    \frac{n_1-|S_1|+|S_2|}{|S_1|+n_2-|S_2|} \leq \frac{(2|S_1| - 2|S_2| + n_2)-|S_1|+|S_2|}{|S_1|+n_2-|S_2|} = 1.
\end{align*}
The proof is now complete. 
\end{proof}

We next show a matching lower bound of $3$ using instances in which all agents are affected by the same facility; we remark that this lower bound has also been shown by \citet{GaiLW22} for the simpler setting with a single facility, and also follows by the work of \citet{chengYZ13}. 

\begin{theorem}\label{thm:general:lower:deterministic}
The approximation ratio of any deterministic strategyproof mechanism is at least $3$.
\end{theorem}

\begin{proof}
Consider an instance $I$ with candidate locations $\{0,2\}$ and two agents that are affected only by $F_1$ and are positioned at $1-\varepsilon$ and $1+\varepsilon$, respectively, where $\varepsilon > 0$ is an infinitesimal; the instance is very similar to that of Figure~\ref{fig:doubleton:lower:randomized:I}. Suppose without loss of generality that the mechanism places $F_1$ at~$0$. 

Consider next an instance $J$ in which agent $i$ that is positioned at $1-\varepsilon$ in $I$ has been moved to $0$; again, the instance is very similar to that of Figure~\ref{fig:doubleton:lower:randomized:J}. Due to strategyproofness, $F_1$ must be placed at $0$ in $J$ as well; otherwise, if $F_1$ is placed at $2$ in $J$, agent $i$ would prefer to misreport its position as $0$ instead of $1-\varepsilon$, to increase its utility from $1-\varepsilon$ to $1+\varepsilon$. The social welfare when $F_1$ is placed at $0$ in $J$ is approximately $1$, whereas the optimal social welfare is approximately $3$ when it is placed at $2$, thus leading to a lower bound of $3$. 
\end{proof}

\subsection{Randomized Mechanisms}\label{sec:general:randomized}
We finally show that an improved upper bound of $2$ can easily be achieved by a simple randomized mechanism that we refer to as {\sc Equiprobable-LR}; note that we already have a lower bound of $3/2$ for any randomized mechanism due to Theorem~\ref{thm:doubleton:lower:randomized}. The mechanism simply chooses each of the solutions $(L,R)$ and $(R,L)$ with probability $1/2$. Clearly, this is a universally strategyproof mechanism.

\begin{theorem}\label{thm:general:upper:randomized}
The approximation ratio of {\sc Equiprobable-LR} is at most $2$.
\end{theorem}

\begin{proof}
Let $\bw$ by the randomized solution computed by the mechanism. By definition, we have
\begin{align*}
    \mathbb{E}[\SW(\bw)]  
    &= \sum_{i \in N_1 \cap N_2} \bigg( d(i,L) + d(i,R) \bigg) + \frac12 \cdot \sum_{i \not\in N_1 \cap N_2} \bigg( d(i,L) + d(i,R) \bigg)
\end{align*}
Next, we bound the social welfare of the optimal solution $\bo$. 
We first consider an agent $i \in N_1 \cap N_2$ and switch between the following two cases:
\begin{itemize}
    \item $i \in [L,R]$. Then, for every $j \in [2]$, $d(i,o_j) \leq d(L,R) = d(i,L) + d(i,R)$, and thus 
    $$u_i(\bo) = d(i,o_1) + d(i,o_2) \leq 2 \bigg( d(i,L) + d(i,R) \bigg).$$
    \item $i < L$; the case $i > R$ is symmetric. Then, for every $j \in [2]$, $d(i,o_j) = d(i,L) + d(L,o_j) \leq d(i,L) + d(i,R)$, and thus again
    $$u_i(\bo) = d(i,o_1) + d(i,o_2) \leq 2 \bigg( d(i,L) + d(i,R) \bigg).$$
\end{itemize}
Hence, 
\begin{align*}
    \sum_{i \in N_1 \cap N_2} u_i(\bo) \leq 2 \cdot \sum_{i \in N_1 \cap N_2} \bigg( d(i,L) + d(i,R) \bigg).
\end{align*}
Next consider an agent $i \in N_1 \setminus N_2$. Since $o_1 \in [L,R]$, using the triangle inequality, we obtain
\begin{align*}
    u_i(\bo) &= d(i,o_1) = \frac12 d(i,o_1) + \frac12 d(i,o_1) \\
    &\leq \frac12 \bigg( d(i,L) + d(L,o_1) \bigg) + \frac12 \bigg( d(i,R) + d(R,o_1) \bigg) \\
    &\leq \frac12 \bigg( d(i,L) + d(i,R) + d(L,R) \bigg) \\
    &\leq d(i,L) + d(i,R).
\end{align*}
Hence, 
\begin{align*}
    \sum_{i \in N_1 \setminus N_2} u_i(\bo) \leq \sum_{i \in N_1 \setminus N_2} \bigg( d(i,L) + d(i,R) \bigg).
\end{align*}
Similarly, for the agents of $N_2 \setminus N_1$, we have
\begin{align*}
    \sum_{i \in N_2 \setminus N_1} u_i(\bo) \leq \sum_{i \in N_2 \setminus N_1} \bigg( d(i,L) + d(i,R) \bigg).
\end{align*}
By putting everything together, we obtain
\begin{align*}
    \SW(\bo) \leq 2 \cdot \sum_{i \in N_1 \cap N_2} \bigg( d(i,L) + d(i,R) \bigg) + \sum_{i \not\in N_1 \cap N_2} \bigg( d(i,L) + d(i,R) \bigg) \leq 2 \cdot \mathbb{E}[\SW(\bw)].  
\end{align*}
The proof is complete. 
\end{proof}

\section{Open Problems}
Our work leaves open some quite challenging questions. In particular, the gap between $1.5$ and $1.522$ for randomized mechanisms in the non-optional setting, and the gap between $1.5$ and $2$ for randomized mechanisms in the general optional setting. Beyond those, it would also be interesting to consider other objective functions besides the social welfare in our optional setting, such as the egalitarian welfare or the sum of square distances that have been studied in some previous work on obnoxious facility location models~\citep{YeMZ15,LiPS23,ZhaoLNF24}. Another interesting direction could be to consider different types of preferences over the facilities besides optional, and also to study whether improvements can be achieved when given predictions about the optimal locations of the facilities, a direction that has gained interest recently~\citep{AgrawalBGOT22,istrate2022obnoxious-predictions,XuL22,FangFLN24,christodoulou2024advice}. Finally, it might make sense to study more constrained settings in which each facility has a cost and we can only build a set of facilities that are within a given budget, similarly to the work of~\citep{deligkas2023limited,FangL24} on limited resources for desirable facilities and optional or fractional preferences.

\bibliographystyle{plainnat}
\bibliography{references}

\end{document}